\documentclass[fleqn,11pt,twoside]{article}

\usepackage{amsthm,amsthm,amssymb, color, xcolor,epsfig, graphics, subfigure}

\usepackage{amsmath, graphicx, latexsym, lscape}

\makeatletter
\newcommand{\copyrightnote}[2]{{\renewcommand{\thefootnote}{}
 \footnotetext{\small\it
\begin{flushleft}
 \copyright \ #1   #2  
\end{flushleft}}}}

\newcommand{\Name}[1]{\begin{flushleft}
                       \LARGE \bf #1
                       \end{flushleft}\vspace{-3mm}}

\newcommand{\Author}[1]{\begin{flushleft}
                       \it #1 \end{flushleft}}

\newcommand{\Address}[1]{\begin{flushleft}
                       \it #1 \end{flushleft}}

\newcommand{\Date}[1]{\begin{flushleft}
                      \small  \it #1 \end{flushleft}}

\newcommand{\evenhead}{Author \ name}
\newcommand{\oddhead}{Article \ name}

\renewcommand{\@evenhead}{
\hspace*{-3pt}\raisebox{-15pt}[\headheight][0pt]{\vbox{\hbox to \textwidth
{\thepage \hfil \evenhead}\vskip4pt \hrule}}}
\renewcommand{\@oddhead}{
\hspace*{-3pt}\raisebox{-15pt}[\headheight][0pt]{\vbox{\hbox to \textwidth
{\oddhead \hfil \thepage}\vskip4pt\hrule}}}
\renewcommand{\@evenfoot}{}
\renewcommand{\@oddfoot}{}

\long\def\@makecaption#1#2{%
  \vskip\abovecaptionskip
  \sbox\@tempboxa{\small \textbf{#1.}\ \ #2}%
  \ifdim \wd\@tempboxa >\hsize
    {\small \textbf{#1.}\ \ #2}\par
  \else
    \global \@minipagefalse
    \hb@xt@\hsize{\hfil\box\@tempboxa\hfil}%
  \fi
  \vskip\belowcaptionskip}

\newcommand{\JNMPnumberwithin}[3][\arabic]{%
  \@ifundefined{c@#2}{\@nocounterr{#2}}{%
    \@ifundefined{c@#3}{\@nocnterr{#3}}{%
      \@addtoreset{#2}{#3}%
      \@xp\xdef\csname the#2\endcsname{%
        \@xp\@nx\csname the#3\endcsname .\@nx#1{#2}}}}%
}

\newcommand{\resetfootnoterule} {
  \renewcommand\footnoterule{%
  \kern-3\p@
  \hrule\@width.4\columnwidth
  \kern2.6\p@}
}

\renewcommand{\footnoterule}{}

\makeatother

\theoremstyle{definition}

\newtheorem{theorem}{Theorem}

\newcommand{\bbR}{{\mathbb R}}

\begin{document}

\renewcommand{\evenhead}{ {\LARGE\textcolor{blue!10!black!40!green}{{\sf \ \ \ ]ocnmp[}}}\strut\hfill 
N Gie{\ss}ing and Yu B Suris
}
\renewcommand{\oddhead}{ {\LARGE\textcolor{blue!10!black!40!green}{{\sf ]ocnmp[}}}\ \ \ \ \  
Conserved quantities of discretizations by polarization
}

\thispagestyle{empty}
\newcommand{\FistPageHead}[3]{
\begin{flushleft}
\raisebox{8mm}[0pt][0pt]
{\footnotesize \sf
\parbox{150mm}{{\textcolor{blue!10!black!40!green}{{\bf Open Communications in Nonlinear Mathematical Physics}}}
\ \ {Special Issue: Hietarinta}, 2026\\[0.1cm]
\strut\hfill 
ocnmp:18664
pp #2\hfill {\sc #3}}}\vspace{-13mm}
\end{flushleft}}

\FistPageHead{1}{\pageref{firstpage}--\pageref{lastpage}}{ \ \ }

\strut\hfill

\strut\hfill

\copyrightnote{The authors. Distributed under a Creative Commons Attribution 4.0 International License}

\begin{center}

{\bf {\large A Special OCNMP Issue in Honour of Jarmo Hietarinta}}\\[0.2cm]
{\bf {\large on the Occasion of his 80th Birthday}}
\end{center}

\smallskip

\Name{Conserved quantities of discretizations by polarization}

\Author{Noah Gie{\ss}ing}

\Address{FIZ Karlsruhe, Hermann-von-Helmholtz-Platz 1, 76344 Eggenstein-Leopoldshafen, Germany, e-mail: {\tt noah.giessing@fiz-karlsruhe.de}}

\Author{Yuri B. Suris}

\Address{Institut f\"ur Mathematik, MA 7-1, 
Technische Universit\"at Berlin, Str. des 17. Juni 136, 10623 Berlin, Germany, 
e-mail: {\tt suris@math.tu-berlin.de}}

\Date{Received June 30, 2026; Accepted July 29, 2026}

\setcounter{equation}{0}

\smallskip

\noindent
{\bf Citation format for this Article:}\newline
Noah Gie{\ss}ing and Yuri B. Suris,
Conserved quantities of discretizations by polarization,
{\it Open Commun. Nonlinear Math. Phys.}, Special Issue:\,Hietarinta, ocnmp:18664, \pageref{firstpage}--\pageref{lastpage}, 2026.

\strut\hfill

\noindent
{\bf The permanent Digital Object Identifier (DOI) for this Article:}\newline
{\it 10.46298/ocnmp.18664}
\strut\hfill

\begin{abstract}

\noindent 
Recently, a family of unconventional integrators for higher order ODEs with polynomial vector fields was proposed, based on the polarization of vector fields. The simplest instance is the by now famous Kahan discretization for first order ODEs with quadratic vector fields. All these integrators possess remarkable conservation properties. In particular, for the first and the second order Hamiltonian ODEs, the discretization by polarization possesses an integral of motion and an invariant volume form. In this note, we extend our previously proposed algebraic approach to derivation of these integrals to discretizations of ODEs of an arbitrary order. For all orders $\ge 3$, these integrals are new.

\end{abstract}

\label{firstpage}


\section{Introduction}

This paper is a further development of \cite{S} and is devoted to integrals of motion of a special numerical scheme for ordinary differential equations of higher order in $\mathbb R^d$ of the type
\begin{equation}\label{eq differential deg m}
x^{(m)}=f(x),
\end{equation}
where all components of the vector field $f$ are polynomials of degree $m+1$. The {\em discretization by polarization} with the stepsize $\epsilon$ of such an equation, as introduced and studied in \cite{HQ, MMQ}, is the following difference equation:
\begin{equation}\label{eq difference deg m}
\Delta^m x_n={\rm pol}_{m+1}f(x_n,\ldots,x_{n+m}).
\end{equation}
Here $\Delta^m$ is the forward difference operator,
\begin{equation}\label{eq difference deg m 0}
\Delta^m x_n=\frac{1}{\epsilon^m}\sum_{j=0}^m (-1)^j\binom{m}{j}x_{n+m-j},
\end{equation}
while ${\rm pol}_{m+1}f(x_n,\ldots,x_{n+m})$ is the {\em polarization} of the degree $m+1$ polynomial $f$. This is a symmetric $(m+1)$-linear form satisfying
${\rm pol}_{m+1} f(x,\ldots,x)=f(x)$, whose precise definition will be discussed in Section \ref{sect polarization}.  Actually, we will consistently omit the index $n$ (or, equivalently, set $n=0$) in difference equations like \eqref{eq differential deg m}, thus abbreviating it to
\begin{equation}\label{eq differential deg m from 0}
\Delta^m x_0=\frac{1}{\epsilon^m}\sum_{j=0}^m (-1)^j\binom{m}{j}x_{m-j}={\rm pol}_{m+1}f(x_0,\ldots,x_{m}).
\end{equation}
This equation is linear with respect to $x_m$, thus can be solved to give a rational expression
\begin{equation}
x_m=\Phi_\epsilon(x_0,\ldots,x_{m-1}).
\end{equation}
Moreover, this equation is symmetric with respect to $j\leftrightarrow m-j$, $\epsilon\to -\epsilon$, therefore the previous relation can be reversed to
\begin{equation}
 x_0=\Phi_{-\epsilon}(x_m,\ldots,x_1).
 \end{equation}
 In other words, the map
 \begin{equation}\label{eq f difference}
\Psi_\epsilon:  \begin{pmatrix} x_0 \\ x_1 \\ \vdots \\ x_{m-1}\end{pmatrix} \mapsto  \begin{pmatrix} x_1 \\ \vdots\\ x_{m-1}  \\ x_{m}\end{pmatrix} =
 \begin{pmatrix} x_1 \\ \vdots \\ x_{m-1} \\ \Phi_\epsilon(x_0,\ldots,x_{m-1})\end{pmatrix} 
 \end{equation}
 is birational, with the inverse map
 $$
\Psi_\epsilon^{-1}:  \begin{pmatrix} x_1 \\ \vdots \\ x_{m-1}  \\ x_{m}\end{pmatrix} \mapsto  \begin{pmatrix} x_0 \\ x_1 \\ \vdots \\ x_{m-1}\end{pmatrix} =
 \begin{pmatrix} \Phi_{-\epsilon}(x_m,\ldots,x_1) \\ x_1 \\ \vdots \\ x_{m-1}\end{pmatrix} .
 $$

The by now famous Kahan discretization \cite{K} is the particular case of discretization by polarization with $m=1$, a one-step numerical method for ODEs in $\bbR^d$,
\begin{equation}\label{eq dyn syst}
\dot x=f(x),
\end{equation}
where all components of the vector field $f$ are polynomials of degree 2. The Kahan discretization with the stepsize $\epsilon$ is the following difference equation:
\begin{equation}\label{eq Kahan}
(x_{1}-x_0)/\epsilon ={\rm pol}_2f(x_0,x_{1}).
\end{equation}
This defines a birational map $x_1=\Psi_\epsilon(x_0)$, with the inverse map $x_0=\Psi_{-\epsilon}(x_1)$.

Similarly, the particular case of discretization by polarization with $m=2$ introduced in \cite{HQ}, is a one-step numerical method for the second order ODEs in $\mathbb R^d$, 
\begin{equation}\label{eq 2nd order}
\ddot x =f(x),
\end{equation}
where all components of $f(x)$ are polynomials of degree 3. The polarization discretization of such an equation with the stepsize $\epsilon$ is the following second order difference equation:
\begin{equation}\label{eq 2nd order discrete}
(x_{2}-2x_1+x_{0})/\epsilon^2={\rm pol}_3f(x_{0},x_1,x_{2}).
\end{equation}
This defines a birational map $(x_1,x_2)=\Psi_\epsilon(x_0,x_1)$, which enjoys the symmetry with respect to $x_{0}\leftrightarrow x_{2}$, $\epsilon\to -\epsilon$. 

For the cases $m=1,2$, remarkable results concerning existence of integrals of motion (and of invariant measures) have been established in \cite{CMOQ1}, resp. in \cite{HQ,MMQ}.

If $m=1$ and $f(x)=K\nabla H(x)$, where $H:\mathbb R^d\to \mathbb R$ is a cubic polynomial and $K\in {\rm so}(d)$ is a non-degenerate skew-symmetric matrix, so that equation \eqref{eq dyn syst} is Hamiltonian, then its Kahan discretization \eqref{eq Kahan} possesses an integral of motion \cite{CMOQ1}. 

Likewise, if $m=2$ and $f(x)=-K\nabla U(x)$, where $U:\mathbb R^d\to\mathbb R$ is a polynomial of degree 4 and $K\in {\rm Symm}(d)$ is a non-degenerate symmetric matrix, then equation \eqref{eq 2nd order} is equivalent to a canonical Hamiltonian system with the Hamilton function $H(x,p)=\frac{1}{2}\langle p, Kp\rangle+U(x)$. Indeed, equations of motion of the latter read $\dot x=Kp$, $\dot p=-\nabla U(x)$.  The discretization by polarization of this system, given by \eqref{eq 2nd order discrete}, possesses an integral of motion \cite{HQ, MMQ}. 

In \cite{S}, a novel derivation and algebraic interpretation of these results has been given. The goal of the present paper is to show that the algebraic mechanism uncovered in \cite{S} is actually valid for any $m$. The integrals thus found have been previously unknown (for $m\ge 3$). An invariant volume form for these discrete systems has been found in \cite{MMQ}.

\section{Polarization}
\label{sect polarization}

For a {\em homogeneous} polynomial $F$ of degree $m+1$, one defines
$$
{\rm pol}_{m+1} F(x_0,\ldots,x_m)=\frac{1}{(m+1)!}\sum_{\stackrel{1\le j\le m+1 }{ 0\le i_1 < \ldots < i_j\le m}} (-1)^{m+1-j}F(x_{i_1}+\ldots+x_{i_j}).
$$
This is the symmetric $(m+1)$-linear form which turns into the form $F(x)$ of degree $m+1$ on the diagonal:
$$
{\rm pol}_{m+1} F(x,\ldots,x)=F(x).
$$
For instance, for a quadratic form $Q(x)$, its polarization is the symmetric bilinear form,
$$
{\rm pol}_2Q(x_0,x_1)=\frac{1}{2}\big(Q(x_0+x_1)-Q(x_0)-Q(x_1)\big).
$$
Similarly, for a cubic form $C(x)$, its polarization is the symmetric trilinear form 
\begin{eqnarray*}
{\rm pol}_3C(x_0,x_1,x_2) & = & \frac{1}{6}\big(C(x_0+x_1+x_2)-C(x_0+x_1)-C(x_0+x_2)-C(x_1+x_2)\\
 & & +C(x_0)+C(x_1)+C(x_2)\big).
\end{eqnarray*}

For a {\em non-homogeneous} polynomial $F(x)$ of degree $m+1$, one extends it to a form of degree $m+1$ in homogeneous coordinates, $\widetilde F(x,z)=z^{m+1} F(x/z)$, computes the $(m+1)$-linear symmetric form ${\rm pol}_{m+1}\widetilde F$ and then sets ${\rm pol}_{m+1} F={\rm pol}_{m+1}\widetilde F|_{z_0=\ldots =z_m=1}$. One can check directly that for a homogeneous polynomial $F_k$ of degree $k\le m$ one has
\begin{equation}\label{eq pol incomplete}
{\rm pol}_{m+1} F_k(x_0,\ldots,x_m)=\frac{1}{\binom{m+1}{k}}\sum_{0\le i_1 < \ldots < i_k \le m} {\rm pol}_k F_k(x_{i_1}, \ldots, x_{i_k}).
\end{equation}
For instance, for a linear form $L(x)$ we obtain 
\begin{eqnarray*}
{\rm pol}_2L(x_0,x_1) & = & \frac{1}{2}\big(L(x_0)+L(x_1)\big),\\
{\rm pol}_3L(x_0,x_1,x_2) & = & \frac{1}{3}\big(L(x_0)+L(x_1)+L(x_2)\big),
\end{eqnarray*}
while for a quadratic form $Q(x)$ we obtain 
$$
{\rm pol}_3Q(x_0,x_1,x_2)=\frac{1}{3}\big({\rm pol}_2Q(x_0,x_1)+{\rm pol}_2Q(x_0,x_2)+{\rm pol}_2Q(x_1,x_2)\big).
$$

\section{Main result}
\label{sect main}

\begin{theorem}\label{th main}
Consider system \eqref{eq differential deg m}, where
\begin{equation}
f(x)=(-1)^{m+1}K\nabla U(x),
\end{equation}
where $U:\mathbb R^d\to\mathbb R$ is a polynomial of degree $m+2$, while $K$ is a non-degenerate $d\times d$ matrix,
\begin{equation}
\left\{\begin{array}{ll} K\in{\rm so}(d), & \; {\rm if}\;\; m\;\;{\rm odd}, \\ 
K\in{\rm Symm}(d), & \; {\rm if}\;\; m\;\;{\rm even}. \end{array}\right.
\end{equation}
Denote $U(x)=\sum_{k=1}^{m+2} U_k(x)$, where $U_k(x)$ are homogeneous polynomials of degree $k$.
Then discretization by polarization \eqref{eq difference deg m 0} possesses a conserved quantity
\begin{equation}\label{eq H eps}
H_\epsilon(x_0,\ldots,x_m)=T_\epsilon(x_0,\ldots,x_m)+V_\epsilon(x_0,\ldots,x_m),
\end{equation}
where
\begin{eqnarray}
T_\epsilon(x_0,\ldots,x_m) & = & \frac{1}{\epsilon^m}\sum_{k=0}^{m} (-1)^k\binom{m}{k}\sum_{j=0} ^ {m-k-1}\langle x_{j},K^{-1}x_{j+k+1}\rangle, \\
V_\epsilon(x_0,\ldots,x_m) & = & \sum_{k=1}^{m+1} (m+2-k)\;{\rm pol}_{m+1}U_{k}(x_0,\ldots,x_m).
\end{eqnarray}
\end{theorem}

{\bf Proof.} 
We give details for the case $m$ even (the case $m$ odd is very similar). The departure point is two formulas. The first is obtained by taking the scalar product of equation \eqref{eq difference deg m 0} with $x_{m+1}$:
\begin{equation}\label{eq aux 1}
\frac{1}{\epsilon^m}\sum_{k=0}^m (-1)^k\binom{m}{k}\langle x_{m+1}, K^{-1} x_{m-k}\rangle=-\sum_{k=1}^{m+2} \langle x_{m+1}, 
{\rm pol}_{m+1}\nabla U_{k}(x_{0},\ldots,x_{m})\rangle.
\end{equation}
The second is obtained by taking the scalar product of the upshifted version of \eqref{eq difference deg m 0} with $x_{0}$:
\begin{equation}\label{eq aux 2}
\frac{1}{\epsilon^m}\sum_{k=0}^m (-1)^k\binom{m}{k}\langle x_{0}, K^{-1} x_{m-k+1}\rangle=-\sum_{k=1}^{m+2} \langle x_{0}, 
{\rm pol}_{m+1}\nabla U_{k}(x_{1},\ldots,x_{m+1})\rangle.
\end{equation}

We start with the difference of their left-hand sides. Changing index in the second one $k\to m-k$ and taking into account the symmetry of the matrix $K^{-1}$, we find:
$$
\frac{1}{\epsilon^m}\sum_{k=0}^m (-1)^k\binom{m}{k}\langle x_{m-k}, K^{-1} x_{m+1}\rangle-\frac{1}{\epsilon^m}\sum_{k=0}^m (-1)^k\binom{m}{k}\langle x_{0}, K^{-1} x_{k+1}\rangle.
$$
Due to telescoping, the difference of left-hand sides of \eqref{eq aux 1} and \eqref{eq aux 2} equals:
$$
\frac{1}{\epsilon^m}\sum_{k=0}^m (-1)^k\binom{m}{k}\sum_{j=1} ^ {m-k}\langle x_{j},K^{-1}x_{j+k+1}\rangle-\frac{1}{\epsilon^m}\sum_{k=0}^m (-1)^k\binom{m}{k}\sum_{j=0} ^ {m-k-1}\langle x_{j},K^{-1}x_{j+k+1}\rangle
$$
\begin{equation}\label{eq aux 3}
=T_\epsilon(x_1,\ldots,x_{m+1})-T_\epsilon(x_0,\ldots,x_m).
\end{equation}

We proceed with the right-hand sides of formulas \eqref{eq aux 1} and \eqref{eq aux 2}. According to \eqref{eq pol incomplete}, the right-hand side of \eqref{eq aux 1} can be written as
$$
-\sum_{k=1}^{m+2}\frac{1}{\binom{m+1}{k-1}}\sum_{0 \le i_1 < \ldots < i_{k-1} \le m}  \langle x_{m+1}, {\rm pol}_{k-1} \nabla U_k(x_{i_1}, \ldots, x_{i_{k-1}})\rangle.
$$
By Euler theorem about homogeneous polynomials, this equals
$$
-\sum_{k=1}^{m+2}\frac{k}{\binom{m+1}{k-1}}\sum_{0 \le i_1 < \ldots < i_{k-1} \le m} {\rm pol}_{k} U_k(x_{i_1}, \ldots, x_{i_{k-1}},x_{m+1})
$$
$$
=-\sum_{k=1}^{m+2}\frac{k}{\binom{m+1}{k-1}}\sum_{0 \le i_1 < \ldots < i_{k-1} < i_{k}=m+1} {\rm pol}_{k} U_{k}(x_{i_1}, \ldots, x_{i_{k-1}},x_{i_{k}}).
$$
Similarly, the right-hand side of \eqref{eq aux 2} can be transformed as follows:
$$
-\sum_{k=1}^{m+2}\frac{1}{\binom{m+1}{k-1}}\sum_{1 \le i_2 < \ldots < i_{k} \le m+1}  \langle x_{0}, {\rm pol}_{k-1} \nabla U_{k}(x_{i_2}, \ldots, x_{i_{k}})\rangle
$$
$$
=-\sum_{k=1}^{m+2}\frac{k}{\binom{m+1}{k-1}}\sum_{1 \le i_2 < \ldots < i_{k} \le m+1}  {\rm pol}_{k} U_{k}(x_{0},x_{i_2}, \ldots, x_{i_{k}})
$$
$$
=-\sum_{k=1}^{m+2}\frac{k}{\binom{m+1}{k-1}}\sum_{0=i_1 < i_2 < \ldots < i_{k} \le m+1}  {\rm pol}_{k} U_{k}(x_{i_1}, x_{i_2}, \ldots, x_{i_k}).
$$
Thus, the difference of the right-hand sides of equations \eqref{eq aux 1} and \eqref{eq aux 2} equals
$$
-\sum_{k=1}^{m+2}\frac{k}{\binom{m+1}{k-1}}\bigg(\sum_{0 \le i_1 < \ldots < i_{k}=m+1} - \sum_{0=i_1 < \ldots < i_{k} \le m+1} \bigg) {\rm pol}_{k} U_{k}(x_{i_1}, x_{i_2}, \ldots, x_{i_k}).
$$
In the interior sum, all terms with $i_1=0$ and $i_{k}=m+1$ cancel away, so it can be represented as 
$$
\sum_{1 \le i_1 < \ldots < i_{k} =m+1} - \sum_{0=i_1 < \ldots < i_{k} \le m}.
$$
We add the vanishing expression
$$
\sum_{1\le i_1 < \ldots < i_{k} < m+1} - \sum_{0 <i_1 < \ldots < i_{k} \le m} =0,
$$
to put the result as 
$$
\sum_{1 \le i_1 < \ldots < i_{k} \le m+1} - \sum_{0 \le i_1 < \ldots < i_{k} \le m} .
$$
In other words, the difference of the right-hand sides of equations \eqref{eq aux 1} and \eqref{eq aux 2} equals
$$
-\sum_{k=1}^{m+2}\frac{k}{\binom{m+1}{k-1}}\bigg(\sum_{1 \le i_1 < \ldots < i_{k}\le m+1} - \sum_{0 \le i_1 < \ldots < i_{k} \le m} \bigg) {\rm pol}_{k} U_{k}(x_{i_1},  \ldots, x_{i_k}).
$$
According to \eqref{eq pol incomplete}, this can be represented as
$$
-\sum_{k=1}^{m+2}\frac{k}{\binom{m+1}{k-1}} \binom{m+1}{k} \Big({\rm pol}_{m+1} U(x_1,\ldots,x_{m+1}) -  {\rm pol}_{m+1} U(x_0,\ldots,x_{m})\Big).
$$
It remains to observe that
$$
\frac{\binom{m+1}{k}}{\binom{m+1}{k-1}}\ k=m+2-k,
$$
so that finally the difference of the right-hand sides of equations \eqref{eq aux 1} and \eqref{eq aux 2} equals
$$
-\sum_{k=1}^{m+2}(m+2-k) \Big({\rm pol}_{m+1} U(x_1,\ldots,x_{m+1}) -  {\rm pol}_{m+1} U(x_0,\ldots,x_{m})\Big)
$$
\begin{equation}\label{eq aux 4}
=-\Big(V_\epsilon(x_1,\ldots,x_{m+1})-V_\epsilon(x_0,\ldots,x_m)\Big).
\end{equation}
Combining \eqref{eq aux 3} and \eqref{eq aux 4}, we finish the proof.
\qed
\medskip

\noindent
{\bf Remark 1.} Usually, an integral of motion for a difference equation of order $m$ is understood as a function $I_\epsilon(x_0,\ldots,x_{m-1})$ invariant under map $\Psi_\epsilon$ in \eqref{eq f difference}, i.e., $I_\epsilon=I_\epsilon\circ \Psi_\epsilon$, which can be also expressed as 
$$
I_\epsilon(x_0,\ldots,x_{m-1})=I_\epsilon(x_1,\ldots,x_m).
$$
We say that a function $H_\epsilon(x_0,\ldots,x_{m-1},x_m)$ is a conserved quantity for a difference equation of order $m$ if
$$
H_\epsilon(x_0,\ldots,x_{m-1},x_m)=H_\epsilon(x_1,\ldots,x_m,x_{m+1})
$$
on solutions of the difference equation, i.e., if the previous formula holds true with $x_m=\Phi_\epsilon(x_0,\ldots,x_{m-1})$ and $x_{m+1}=\Phi_\epsilon(x_1,\ldots,x_m)$. Of course, upon this substitution the resulting expression is an integral of motion 
$$
I_\epsilon(x_0,\ldots,x_{m-1})=H_\epsilon\big(x_0,\ldots,x_{m-1},\Phi_\epsilon(x_0,\ldots,x_{m-1})\big),
$$ 
however it becomes much more algebraically complicated.
\medskip

\noindent
{\bf Remark 2.} The conserved quantity simplifies drastically, if the polynomial $U(x)$ is homogeneous of degree $m+2$. Indeed, then $V_\epsilon=0$, and $H_\epsilon=T_\epsilon$. In principle, the inhomogeneous case could be reduced to the homogeneous one by introducing an extra dependent variable. Our proof shows that the extra variable can be explicitly and cleanly eliminated.

\section{Examples}
\label{sect examples}

\paragraph{{\bf Case} {\boldmath${m=1}$.}}
Consider the system
\begin{equation}\label{eq Ham system}
\dot{x}=K\nabla U, 
\end{equation}
where $U:\bbR^{d}\to \bbR$ is a polynomial of degree 3, and $K$ is a non-degenerate skew-symmetric $d\times d$ matrix (so that the dimension $d$ is even). This is a canonical Hamiltonian system, and the Hamilton function $U(x)$ is an integral of motion. The right-hand side of equation \eqref{eq Ham system} is of degree 2, and the corresponding discretization by polarization is nothing but the Kahan discretization, see \eqref{eq Kahan}. According to Theorem \ref{th main}, difference equation \eqref{eq Kahan} possesses a conserved quantity \eqref{eq H eps} with
\begin{eqnarray}
T_\epsilon(x_0,x_1) & = & \frac{1}{\epsilon}\langle x_0,K^{-1}x_1\rangle, \\
V_\epsilon(x_0,x_1) & = & {\rm pol}_2U_2(x_0,x_1)+2\,{\rm pol}_2U_1(x_0,x_1). \qquad
\end{eqnarray}
This conserved quantity was found in \cite{S}. It turns into an integral of motion found in \cite{CMOQ1} upon substitution $x_1=\Psi_\epsilon(x_0)$.
Its continuous time limit is obtained by setting $x_0=x$, $x_1=x+\epsilon x+O(\epsilon^2)$ and then sending $\epsilon\to 0$:
\begin{eqnarray}
T_\epsilon(x_0,x_1) & \to & T_0(x,\dot x)=\langle x, K^{-1}\dot x\rangle, \label{eq m=1 T cont} \\
V_\epsilon(x_0,x_1) & \to &  V_0(x)=U_2(x)+2U_1(x). \label{eq m=1 V cont}
\end{eqnarray}
The resulting conserved quantity is unusual, as it contains $\dot{x}$. To put it in the usual form, one should use equations of motion \eqref{eq Ham system}. 
By virtue of \eqref{eq Ham system} and of Euler theorem on homogeneous functions, we find the following expression for the function \eqref{eq m=1 T cont}:
$$
 T_0(x) = \langle x, \nabla U(x) \rangle= 3U_3(x)+2U_2(x)+U_1(x).
$$
Upon adding \eqref{eq m=1 V cont}, we end up with the integral of motion $3U(x)$. 
\medskip

\paragraph{{\bf Case} {\boldmath${m=2}$.}}
Consider the system
\begin{equation}\label{eq Ham system 2nd order}
\ddot{x}=-K\nabla U, 
\end{equation}
where $U:\bbR^{d}\to \bbR$ is a polynomial of degree 4, and $K$ is a non-degenerate symmetric $d\times d$ matrix. This system is Lagrangian and admits an integral of motion
\begin{equation}\label{eq int 2nd order}
H(x,\dot x)=\frac{1}{2}\langle \dot x, K^{-1}\dot x\rangle +U(x).
\end{equation}
The right-hand side of equation \eqref{eq Ham system 2nd order} is of degree 3, and we consider the corresponding discretization by polarization, see \eqref{eq 2nd order discrete}. According to Theorem \ref{th main}, difference equation \eqref{eq 2nd order discrete} possesses a conserved quantity \eqref{eq H eps} with
\begin{eqnarray}
\epsilon^2 T_\epsilon(x_0,x_1,x_2) & = & \langle x_0,K^{-1}x_1\rangle+\langle x_1,K^{-1}x_2\rangle-2\langle x_0,K^{-1}x_2\rangle, \\
V_\epsilon(x_0,x_1,x_2) & = & {\rm pol}_3U_3(x_0,x_1,x_2)+2\,{\rm pol}_3U_2(x_0,x_1,x_2)+3\,{\rm pol}_3U_1(x_0,x_1,x_2). \qquad
\end{eqnarray}
This conserved quantity was found in \cite{S}. It turns into an integral of motion found in \cite{HQ} upon substitution $x_2=\Psi_\epsilon(x_0,x_1)$. It is instructive to look at the continuous time limit of the conserved quantity $H_\epsilon$. To perform this, we set $x_k=x+(k\epsilon)\dot{x}+\tfrac{1}{2}(k\epsilon)^2\ddot x+O(\epsilon^3)$, and then send $\epsilon\to 0$. We find:
\begin{eqnarray}
T_\epsilon(x_0,x_1,x_2) & \to & T_0(x,\dot x, \ddot x)=2\langle \dot x, K^{-1}\dot x\rangle-\langle x, K^{-1}\ddot x \rangle, \label{eq m=2 T cont} \\
V_\epsilon(x_0,x_1,x_2) & \to &  V_0(x)=U_3(x)+2U_2(x)+3U_1(x). \label{eq m=2 V cont}
\end{eqnarray}
Thus, we arrive at an unusual conserved quantity, as it contains $\ddot{x}$. To put it in the usual form, one should use equations of motion \eqref{eq Ham system 2nd order}.  By virtue of \eqref{eq Ham system 2nd order} and of Euler theorem on homogeneous functions, we find the following expression for the function \eqref{eq m=2 T cont}:
\begin{eqnarray*}
 T_0(x,\dot x, \ddot x) & = & 
 2\langle \dot x, K^{-1}\dot x\rangle+\langle x, \nabla U(x) \rangle\\
& = & 2\langle \dot x, K^{-1}\dot x\rangle+\big(4 U_4(x)+3U_3(x)+2U_2(x)+U_1(x)\big).
\end{eqnarray*}
Upon adding \eqref{eq m=2 V cont}, we end up with the integral of motion $4H(x,\dot x)$, see \eqref{eq int 2nd order}.
\medskip

\paragraph{{\bf Case} {\boldmath${m=3}$.}}
Consider the system
\begin{equation}\label{eq system 3rd order}
\dddot x=K\nabla U, 
\end{equation}
where $U:\bbR^{d}\to \bbR$ is a polynomial of degree 5, and $K$ is a non-degenerate skew-symmetric $d\times d$ matrix. It is a matter of a straightforward verification that system \eqref{eq system 3rd order} has the following integral of motion:
\begin{equation}\label{eq int 3rd order}
H(x,\dot x,\ddot x)=-\langle \dot x,K^{-1}\ddot x\rangle+U(x). 
\end{equation}
According to Theorem \ref{th main}, the discretization of \eqref{eq system 3rd order} by polarization possesses  a conserved quantity \eqref{eq H eps} with
\begin{eqnarray}
\epsilon^3 T_\epsilon(x_0,\ldots,x_3) & = & \langle x_0,K^{-1}x_1\rangle+\langle x_1,K^{-1}x_2\rangle+\langle x_2,K^{-1}x_3\rangle \nonumber\\
& & -3\big(\langle x_0,K^{-1}x_2\rangle+\langle x_1,K^{-1}x_3\rangle\big) \nonumber\\
& & +3\langle x_0,K^{-1}x_3\rangle,\\
V_\epsilon(x_0,\ldots,x_3) & = & \sum_{k=1}^4(5-k){\rm pol}_kU_k(x_0,\ldots,x_3).
\end{eqnarray}
For the continuous time limit, we set $x_k=\sum_{j=0}^3 \tfrac{(k\epsilon)^j}{j!}x^{(j)}+O(\epsilon^4)$, and then send $\epsilon\to 0$. We find:
\begin{eqnarray}
T_\epsilon(x_0,\ldots,x_3) & \to & T_0(x,\dot x,\ddot x,\dddot x)=-5\langle \dot x, K^{-1}\ddot x\rangle
+\langle x,K^{-1}\dddot x\rangle,   \qquad \label{eq m=3 T cont} \\
V_\epsilon(x_0,\ldots,x_3) & \to &  V_0(x)=\sum_{k=1}^4 (5-k) U_k(x). \label{eq m=3 V cont}
\end{eqnarray}
The put the resulting conserved quantity into the usual form, use equations of motion \eqref{eq system 3rd order}, and by virtue of Euler theorem on homogeneous functions, we find the following expression for the function \eqref{eq m=3 T cont}:
\begin{eqnarray*}
 T_0(x,\dot x, \ddot x,\dddot x) & = & -5\langle \dot x, K^{-1}\ddot x\rangle+\langle x, \nabla U(x) \rangle\\
& = & -5\langle \dot x, K^{-1}\ddot x\rangle+\sum_{k=1}^5 kU_k(x).
\end{eqnarray*}
Upon adding \eqref{eq m=3 V cont}, we end up with the integral of motion $5H(x,\dot x,\ddot x)$, see \eqref{eq int 3rd order}.
\medskip

\paragraph{{\bf Case} {\boldmath${m=4}$.}}
Consider the system
\begin{equation}\label{eq system 4th order}
x^{(4)}=-K\nabla U, 
\end{equation}
where $U:\bbR^{d}\to \bbR$ is a polynomial of degree 6, and $K$ is a non-degenerate symmetric $d\times d$ matrix. This is a Lagrangian system, 
and as such it admits an integral of motion
\begin{equation}\label{eq int 4th order}
H(x,\dot x,\ddot x,\dddot x)=\langle \dot x,K^{-1}\dddot x\rangle-\frac{1}{2}\langle \ddot x,K^{-1}\ddot x\rangle+U(x). 
\end{equation}
According to Theorem \ref{th main}, the discretization of \eqref{eq system 4th order} by polarization possesses  a conserved quantity \eqref{eq H eps} with
\begin{eqnarray}
\epsilon^4 T_\epsilon(x_0,\ldots,x_4) & = & \langle x_0,K^{-1}x_1\rangle+\langle x_1,K^{-1}x_2\rangle+\langle x_2,K^{-1}x_3\rangle
+\langle x_3,K^{-1}x_4\rangle \nonumber\\
& & -4\big(\langle x_0,K^{-1}x_2\rangle+\langle x_1,K^{-1}x_3\rangle+\langle x_2,K^{-1}x_4\rangle\big) \nonumber\\
& & +6\big(\langle x_0,K^{-1}x_3\rangle+\langle x_1,K^{-1}x_4\rangle\big) \nonumber\\
&& -4\langle x_0,K^{-1}x_4\rangle, \\
V_\epsilon(x_0,\ldots,x_4) & = & \sum_{k=1}^5(6-k){\rm pol}_kU_k(x_0,\ldots,x_4).
\end{eqnarray}
In the continuous time limit $\epsilon\to 0$, setting $x_k=\sum_{j=0}^4 \tfrac{(k\epsilon)^j}{j!}x^{(j)}+O(\epsilon^5)$, we find:
\begin{eqnarray}
T_\epsilon(x_0,\ldots,x_4) & \to & T_0(x,\dot x,\ldots,x^{(4)})=-3\langle \ddot x, K^{-1}\ddot x\rangle+6\langle \dot x, K^{-1}\dddot x \rangle
-\langle x,K^{-1}x^{(4)}\rangle,   \qquad \label{eq m=4 T cont} \\
V_\epsilon(x_0,\ldots,x_4) & \to &  V_0(x)=\sum_{k=1}^5 (6-k) U_k(x). \label{eq m=4 V cont}
\end{eqnarray}
To put this unusual conserved quantity, containing the highest derivative $x^{(4)}$, into the usual form, we use equations of motion \eqref{eq system 4th order}, and by virtue of Euler theorem on homogeneous functions, we find the following expression for the function \eqref{eq m=4 T cont}:
\begin{eqnarray*}
 T_0(x,\dot x, \ddot x,\dddot x,x^{(4)}) & = & -3\langle \ddot x, K^{-1}\ddot x\rangle+6\langle \dot x, K^{-1}\dddot x \rangle+\langle x, \nabla U(x) \rangle\\
& = & -3\langle \ddot x, K^{-1}\ddot x\rangle+6\langle \dot x, K^{-1}\dddot x \rangle+\sum_{k=1}^6 kU_k(x).
\end{eqnarray*}
Upon adding \eqref{eq m=2 V cont}, we end up with the integral of motion $6H(x,\dot x,\ddot x,\dddot x)$, see \eqref{eq int 4th order}.

\section{Conclusion}
We hope to have demonstrated the huge potential of the algebraic approach to derivation of conserved quantities for discrete time systems. We stress that the classical notion of an integral of motion should be augmented by alternative, non-standard concepts. In the present paper, this is the concept of a {\em conserved quantity depending on the number of iterates  higher than the order of the underlying difference equation}. A similarly successful notion (which is also much less popular than it deserves to be) is the device of {\em Hirota-Kimura bases}, compare \cite{PPS1, PPS2, PS}. The appeal to the integrable systems community is: there is still much more to be discovered even in the most classical areas!

\subsection*{Acknowledgements}

This paper is based on a part of a master thesis in Mathematics by the first author  supervised at the Technische Universit\"at Berlin by the second author. The first author is supported by the Deutsche Forschungsgemeinschaft (DFG), project number 460135501, NFDI 29/1 ``MaRDI - Mathematische Forschungsdateninitiative''.

\label{lastpage}
\end{document}